\documentclass{ecai} 

\usepackage{latexsym}
\usepackage{amssymb}
\usepackage{amsmath}
\usepackage{amsthm}
\usepackage{booktabs}
\usepackage{enumitem}
\usepackage{graphicx}
\usepackage{color}

\usepackage[hidelinks]{hyperref}
\usepackage[utf8]{inputenc}
\usepackage{algorithm}
\usepackage{algorithmic}
\usepackage[switch]{lineno}
\usepackage{units}
\usepackage{pifont}
\usepackage{complexity}
\usepackage[dvipsnames]{xcolor}

\newtheorem{theorem}{Theorem}

\newtheorem{corollary}[theorem]{Corollary}
\newtheorem{proposition}[theorem]{Proposition}

\newtheorem{definition}{Definition}

\newtheorem{example}{Example}

\newtheorem{myex}{Example} 
\renewenvironment{example}{\begin{myex}\rm}{\hfill$\vartriangle$\end{myex}}
\newenvironment{namedexample}[1][\unskip]{\begin{myex}[#1]\rm}{\hfill$\vartriangle$\end{myex}}

\renewcommand{\paragraph}[1]{\vspace{5pt}\noindent\textbf{#1}}

\newcommand{\BibTeX}{B\kern-.05em{\sc i\kern-.025em b}\kern-.08em\TeX}

\newenvironment{newenv}{\color{black}}{}
\newcommand{\new}[1]{\textcolor{black}{#1}}
\newcommand{\notnew}[1]{\textcolor{black}{#1}}

\newcommand{\Pow}[1]{\mathcal{P}({#1})}

\newcommand{\vecPos}[2]{#2(#1)}
\newcommand{\rep}[2]{r_{#2}(#1)}
\newcommand{\ratio}{\textit{ratio}}

\newcommand{\DHondt}{\text{D'Hondt}_{\omega}}
\newcommand{\Adams}{\text{Adams}_{\omega}}
\newcommand{\cmark}{\ding{51}}%
\newcommand{\xmark}{\ding{55}}%
\newcommand{\prof}[1]{\boldsymbol{#1}}
\newcommand{\quota}{q}
\newcommand{\WLQo}{\text{WLQ}^o}
\newcommand{\WUQo}{\text{WUQ}^o}
\renewcommand{\vec}[1]{\boldsymbol{#1}}

\begin{document}


\begin{frontmatter}


\paperid{8822} 


\title{Apportionment with Weighted Seats}

\author[A]{\fnms{Julian}~\snm{Chingoma}}
\author[A]{\fnms{Ulle}~\snm{Endriss}}
\author[A]{\fnms{Ronald}~\snm{de Haan}}
\author[B]{\fnms{Adrian}~\snm{Haret}}
\author[C]{\fnms{Jan}~\snm{Maly}} 

\address[A]{ILLC, University of Amsterdam}
\address[B]{MCMP, LMU Munich}
\address[C]{DPKM, WU Vienna University of Economics and Business and DBAI, TU Wien}

\begin{abstract}
Apportionment is the task of assigning resources to entities with different entitlements in a fair manner, and specifically a manner that is as proportional as possible. The best-known application is the assignment of parliamentary seats to political parties based on their share in the popular vote. Here we enrich the standard model of apportionment by associating each seat with a weight representing the (objective) value of that seat. A seat's weight reflects the fact that different seats might come with different roles, such as chair or treasurer. We define several apportionment methods and natural fairness requirements for this new setting, and we study the extent to which our methods satisfy these requirements. Our findings show that full fairness is harder to achieve than in the standard apportionment setting. Yet, for several natural relaxations of those requirements we can achieve stronger results than in the more expressive model of fair division with entitlements, where the values of objects are subjective.
\end{abstract}

\end{frontmatter}


\section{Introduction}

Allocating resources in a proportional manner to entities with different entitlements, 
also known as \emph{apportionment}, 
is a core problem of social choice \cite{Balinski2005-xu}:
in federal systems (e.g., the US), states 
receive seats in parliament based on their populations,
while in proportional representation systems (e.g., the Netherlands),
parties receive seats based on their share in the popular vote.
\new{Elsewhere}, the need for apportionment arises 
in the context of fair allocation~\cite{ChakrabortySS21}, the presentation of
statistics \cite{Balinski93} and the handling of bankruptcies \cite{Csoka18}.
But in \new{line} with the paradigmatic example of apportionment, \new{for} the rest of the paper 
we stick to the terminology of seats being assigned to parties.

The merits of different apportionment methods are well understood due to 
an elegant mathematical theory developed for the political realm
\cite{BalinskiYoung1982,Pukelsheim2014}.
However, existing work 
remains limited by the assumption---often 
not met in practice---that all seats are of equal value. 
In this paper, we put forward an enriched model in which seats may have 
different weights reflecting \new{their (objective)} values.


There are numerous scenarios that fit this richer model,
including the distribution of non-liquid assets in bankruptcies 
to beneficiaries with different entitlements,
the assignment of positions on news websites to editorial domains (such as politics, business, or sports), 
based on the readership's relative levels of interest in those domains, and 
the way special-purpose committees are constituted in the \textit{Bundestag}, 
Germany's national parliament \cite{Bundestag2023}, a setting of particular interest that we will discuss in more detail later.

In our enriched model we associate each seat with a weight representing its objective value,
and approach proportionality through the lens of the total weight available.
We find that generalisations of the two central proportionality axioms of apportionment, \emph{lower} and \emph{upper quota}, are impossible to satisfy in general,
but that apportionment methods that faithfully extend well-known standard methods satisfy natural relaxations
of these axioms.
The relaxations, based on the concept of satisfaction `\emph{up to one}' and `\emph{up to any}' seat,
utilise an idea commonly seen in fair division \cite{Budish11,CaragiannisETAlTEAC2019}, 
that has also made its way to participatory budgeting \cite{pierczynskiSP21prop-PB-additive,brillFLMP2022prop-approval-PB}.
Additionally, we study \emph{envy-freeness}, one of the central axioms
in fair division \cite{AmanatidisBFV22},
which turns out to be related to upper quota.
We show that envy-freeness up to any seat, an axiom that
is not always achievable in the general fair division setting, is satisfiable in our setting.
Finally, we find that a direct generalisation of the \emph{house monotonicity} axiom
is prohibitively demanding, but weaker forms are readily satisfied by weighted counterparts of well-known apportionment methods.

\paragraph{Related work.}
Our aim is to assign seats of different weights to parties, similarly 
to how goods are assigned to agents in fair division \cite{AmanatidisBFV22}.
Fair division is a more general problem, because agents in fair division are allowed to have different valuations for the goods, while in our model the weights of the seats are the same for all parties.
There is also a subtle difference of focus: in much of fair division---and certainly the part of the literature considering relaxations of classical fairness notions---it is assumed that all agents deserve the same utility; in apportionment the central concerns stem directly from the fact that parties may have different entitlements.
This divide is bridged by a growing literature on fair
division for agents with different entitlements 
\cite{FarhadiGHLPSSY2019fairAllo,Aziz2020,BabaioffEF21,suksompong2024weighted},
which studies the possibility of achieving proportional allocations for
agents with (positive) cardinal utilities. 
This model of weighted fair division can be seen as a generalisation of our model of apportionment with weighted seats.\footnote{Note that in the former setting, `weights' refer to the weights of agents (i.e., their entitlements) whereas we use it to refer to the weights of seats.} 

Most closely related to our work is a paper by \citet{ChakrabortySS21}, 
to which we will refer often, that uses apportionment methods
to produce picking sequences guaranteeing fair allocations in a fair division setting.
This model can represent a wider range of scenarios than ours, albeit at a cost. 
The increased generality obtained by allowing agents to have different valuations 
makes it much harder to achieve proportionality. 
Indeed, we will see that the fairness guarantees provided by the fair division literature are 
significantly weaker than the ones achievable in our apportionment setting.

Let us now briefly review some further related lines of research from the computational social choice literature.
The first consists of a family of models extending the standard model of apportionment,
be it by allowing voters to cast ballots for more than one party \cite{BrillGPSW20},
by allowing voters to rank the parties \cite{AiriauACKLP23},
or by taking a long-term perspective on apportionment
\cite{Harrenstein22}.
More broadly, approval-based multiwinner voting \cite{LacknerS23} and related models \cite[e.g.,][]{chandakGP24,masarikPS2024}
can be seen as generalisations of apportionment \cite{BrillLS17}.
However, all of these works still assume that all seats are of equal value.

Participatory budgeting \cite{Rey23} 
generalises multiwinner voting by adding weights (candidates become projects with attached costs). 
This may look similar to the generalisation we propose here, but there are crucial differences: in participatory budgeting weights are attached to party members rather than seats; and seats in our setting are \emph{private} (each is assigned to a single stakeholder) whereas funded projects in participatory budgeting are \emph{public} (each is shared by all stakeholders). We will see later that this difference also manifests itself in the normative properties one can achieve.

Finally, our model can be seen as a special case of the public-deci\-sions model of \citet{conitzerF017} when each issue's alternatives are exactly the parties, and a party's utility for a seat is that seat's weight.


\paragraph{Contribution.}
We introduce a new model of apportionment with weighted seats that---in terms of expressivity---is located somewhere between the standard model (where seats are unweighted) and a previously studied model of fair division with entitlements.

The main take-away of our analysis is that obtaining good solutions to the apportionment problem is harder in the presence of weights, but positive results nonetheless are achievable 
for mild relaxations of standard axiomatic properties when appropriately generalised to our richer model.
(The reader will find an overview of our main axiomatic results in Table~\ref{tab:axiomTable} near the end of the paper.)
We stress that we obtain stronger guarantees regarding the satisfaction of axioms than what is possible for the weighted fair division setting with subjective weights. We do so by using the objective weights to define weighted variants of several apportionment methods. 


\paragraph{Outline.} 
Section~\ref{sec:model} introduces the weighted apportionment model.
Sections~\ref{sec:lower-quota} and~\ref{sec:upper-quota} 
offer an axiomatic and a computational analysis of lower and upper quota properties, 
while
Section~\ref{sec:house_mon} studies house monotonicity.
Section~\ref{sec:bundestag_study} presents a case study with real-world data.
Omitted proofs can be found in the Appendix.

\section{The Model}\label{sec:model}

We write $\Pow{U}$ for the sets of all subsets of a set \(U\), and $[k]$ for the set $\{1,\ldots,k\}$, where $k$ is a positive integer.

In our 
model there are $n$ \emph{voters}, 
each voting for one of $m$ \emph{parties}, and
the goal is to fill $k$ \emph{seats} of varying \new{(objective)}
value with members of those parties, based on the votes.%

We make two mild assumptions throughout. 
First, each party is approved by at least one voter. 
Second, to use a term borrowed from \citet{langS2018multi-attribute}, 
we assume \emph{full supply}: 
each party has at least $k$ members, and is thus able to fill all available seats by itself.

A \emph{vote vector} $\prof{v} = (v_{1},\ldots,v_{m}) \in \new{[n]^m}$, 
with \(v_1 + \dots + v_m = n\),
specifies how many votes (out of the total number~\(n\)) each party $p \in [m]$ garnered.
Each seat $t \in [k]$ is associated with a weight $w_t$ indicating how valuable \(t\) is. 
The $k$ seats to be filled are described 
by the \emph{weight vector} $\vec{w} = (w_{1},\ldots,w_{k})\in \mathbb{N}_{\geqslant 1}^{k}$, 
which lists the weights of the seats in non-increasing order. 
The \emph{total weight} is $\omega = \sum_{t\in[k]} w_{t}$. 
A \emph{seat assignment} is a vector $\vec{s} = (s_{1},\ldots,s_{k}) \in [m]^{k}$, 
where $s_{t}=p$ means that party $p\in [m]$ is assigned seat $t \in [k]$ with weight $w_{t}$. 
Given a seat assignment $\vec{s}$, we write $\vecPos{p}{\vec{s}} = (t)_{s_{t} = p}$ for the vector of seats, 
in increasing order of index, assigned to party~$p$ under seat assignment~$\vec{s}$.
An \emph{election instance} is a pair $(\prof{v}, \vec{w})$ of a vote vector~$\prof{v}$ and a weight vector~$\vec{w}$. 
We speak of a \emph{unit-weight} instance in case $w_{t} = w_{t'}$ for all $t,t' \in [k]$.\footnote{The standard model of apportionment deals with such unit-weight instances.}

The core question of apportionment is how to distribute the available seats to parties in a \emph{proportional} manner. 
This is typically formalised in terms of a party's \emph{quota}, i.e.,
the proportion of seats each party is entitled to. 
We do the same in our weighted setting, with the important caveat that the quota in this case 
is construed in terms of the total weight:
the \emph{quota of party} $p$ is defined as $\quota(p) = \omega\cdot\nicefrac{v_p}{n}$.

To judge whether a party satisfies its quota, 
we need to reason about the weights of the seats it obtained.
This leads us to the notion of \emph{representation}. 
Formally, the \emph{representation of party $p$} derived 
from seat assignment $\vec{s}$ is $\rep{p}{\vec{s}} = \sum_{t\in \vecPos{p}{\vec{s}}} w_{t}$, 
i.e., the sum of the weights of the seats assigned to \(p\) according to \(\vec{s}\).
For a weight vector $\vec{w}$, the set of all possible representation values a party can obtain from occupying up to $h \in [k]$ seats can be computed as follows: 
\[
        R(\vec{w})_{[h]} = \left\{\,\sum_{t \in T} w_t \mid T \in \Pow{[k]}\ \text{with}\ |T| \leqslant h \,\right\}.
\]


\noindent
\new{We now turn to} the methods \new{used} to assign seats to parties.
A \emph{weighted-seat assignment method} (WSAM) $F$ takes an election instance $(\prof{v}, \vec{w})$ as input 
and maps it to a winning seat assignment $F(\prof{v}, \vec{w})$.
We focus on two types of WSAMs, which generalise the most prominent methods for standard apportionment 
\cite{BalinskiYoung1982}. 

\begin{definition}[Divisor methods]
	Fix a function $f:\mathbb{R}\times\mathbb{R}\to\mathbb{R}$.
    Given an election instance $(\prof{v},\vec{w})$, 
    the divisor method for $f$ works in $k$ rounds. 
    In round $t\in [k]$, seat~$t$ is given
    to the party~$p$ maximising: 
    \[
        \ratio_p =
        \begin{cases}
            \frac{v_{p}}{f(g_{p}(t),w_{t})} & \text{if } f(g_{p}(t),w_{t}) \neq 0 \\
            \infty & \text{if } f(g_{p}(t),w_{t}) = 0,
        \end{cases}
    \]
    where $g_{p}(t)$ is the sum of the weights of the seats given
    to party~$p$ in earlier rounds.
    If required, a tie-breaking rule is used to choose between parties with equal ratio. 
\end{definition}

\noindent
Intuitively, divisor methods allocate the available seats sequentially, starting with the
most valuable seat, and based on the ratio between $v_{p}$ and $f(g_{p}(t),w_{t})$.
Note that the use of \(g_p(t)\), i.e., the total \emph{weight} (rather than the number) 
of seats assigned to party \(p\)  by round \(t\) reflects our strategy for generalising apportionment 
to the weighted setting.  
It is, of course, possible to allocate the seats in a different (fixed) order but, 
to anticipate results to come, starting with the most valuable seat leads to particularly nice axiomatic properties.

Only certain choices for the function $f$ lead to reasonable divisor methods. 
In the unit-weight apportionment setting, it is common to set $f(g_{p}(t),w_{t})$ to 
$g_p(t)$ (Adams), $g_p(t) + 0.5$ (Sainte-Lagu\"{e}), or $g_p(t) + 1$ (D'Hondt).\footnote{Sainte-Lagu\"{e}, and D'Hondt are also known as \emph{Webster} and \emph{Jefferson}.}
As our focus is on upper and lower quota, 
we narrow attention to Adams, the unique divisor method satisfying upper
quota, and D'Hondt, the unique divisor method satisfying lower quota \cite{Balinski1975-mn}. These rules can be generalised to our setting as follows. 

\begin{definition}[$\Adams$ and $\DHondt$]
    $\Adams$ is the divisor method defined by $f(g_p(t),w_t) = g_p(t)$ and 
    $\DHondt$ is the divisor method defined by $f(g_p(t),w_t) = g_p(t) + w_t$.
\end{definition}

\noindent 
\new{Second on our list, the largest remainder method} (LRM)
assigns each party their lower quota of seats,
as defined below, and then assigns the remaining seats based on the fractional remainder
of each party's quota.\footnote{LRM is also known as the \emph{Hamilton} method in the literature.} But as we will see (Proposition~\ref{prop:WLQ_may_not_exist}), this would not work in the weighted setting. Instead, we put forward the following procedure.

\begin{definition}[Greedy Method]
   In each round $t\in[k]$, the seat $t$ with weight $w_{t}\in\vec{w}$ is assigned to the party~$p$ for which 
   $\quota(p) - g_{p}(t)$ is maximal, with ties broken arbitrarily \new{whenever they arise}. 
\end{definition}

\noindent
Without weights, the Greedy method reduces to~LRM. 

\begin{example}
    \new{Consider three parties obtaining votes} $\prof{v} = (60,30,10)$
    and four seats 
    \new{of weights}
    $\vec{w} = (10, 6, 4, 2)$
    waiting to be filled. 
    \(\Adams\) maximises the ratio \(\nicefrac{v_p}{g_p(t)}\). 
    Since \(g_p(t) = 0\) before a party receives any seats, each party gets a seat
    after the first three rounds; assume tie-breaking assigns party \(1\), \(2\) and \(3\)
    seat \(1\), \(2\) and \(3\), respectively, for the partial assignment \(\vec{s} = (1,2,3,\_)\).
    At round \(t=4\), \(\ratio_p\) is maximised by party \(1\), with \(\ratio_1 = \nicefrac{60}{10}\) 
    versus \(\ratio_2 = \nicefrac{30}{6}\) and \(\ratio_3 = \nicefrac{10}{4}\).
    The final assignment is \(\vec{s} = (1,2,3,1)\).
    \(\DHondt\) maximises the ratio \(\nicefrac{v_p}{(g_p(t)+w_t)}\).
    Assigning the first seat to party \(1\) gives a ratio of \(\nicefrac{60}{(0+10)}\),
    versus \(\nicefrac{30}{(0+10)}\) and \(\nicefrac{10}{(0+10)}\) for parties \(2\) and \(3\), respectively,
    so this seat goes to party \(1\). The second seat goes to party \(2\). For the third seat
    we calculate \(\ratio_1=\nicefrac{60}{(10+4)}\), \(\ratio_2=\nicefrac{30}{(6+4)}\) and 
    \(\ratio_3=\nicefrac{10}{(0+4)}\), so this seat goes to party \(1\). The final assignment is
    \(\vec{s} = (1,2,1,3)\).
    For the Greedy method, the quotas $\quota(p) = \omega\cdot\nicefrac{v_{p}}{n}$ are $\quota(1) = 13.2$, $\quota(2) = 6.6$, and $\quota(3) = 2.2$. So Greedy returns the seat assignment \(\vec{s} = (1,2,1,3)\), just as \(\DHondt\).
\end{example}

\noindent
Note that the WSAMs defined above use the total weight of the seats to determine the assignment, 
in contrast to \citet{ChakrabortySS21} who use 
the number of seats assigned to each party.
Thus, the methods of \citeauthor{ChakrabortySS21} do not directly generalise our WSAMs.

Regarding terminology, while the axioms \new{that} follow are defined as properties of seat assignments, we say that a WSAM $F$ \emph{satisfies property $\mathcal{P}$} if for every election instance $(\prof{v},\vec{w})$ it is the case that every seat assignment $\vec{s}\in F(\prof{v},\vec{w})$ satisfies property~$\mathcal{P}$.


\section{Lower Quota}\label{sec:lower-quota}

In the standard apportionment setting, a perfectly proportional allocation 
would give each party \(p\) 
the share of seats that corresponds precisely to its vote share.
Since there is no guarantee that this share, calculated as $k\cdot\nicefrac{v_p}{n}$, is an integer, 
the immediate fallback is a \emph{lower quota} axiom stating that each party
$p$ should receive \emph{at least} $\lfloor k\cdot\nicefrac{v_{p}}{n} \rfloor$ seats
\cite{BalinskiYoung1982,Pukelsheim2014}. 
For our weighted-seat setting it would thus be natural to define the
\emph{weighted lower quota} 
as
$\lfloor \omega\cdot\nicefrac{v_{p}}{n}\rfloor$.
However, the following example shows that such a quota is not guaranteed to be achievable,
even in the simplest case of two parties and two seats.

\begin{example}
    Consider two parties with $\prof{v} = (1,1)$
    and $\vec{w} = (99,1)$. So $\lfloor \omega\cdot\nicefrac{v_{p}}{n}\rfloor = 50$ for both parties $p\in [2]$, but there is no way for both to receive seats with weight at least \(50\).
\end{example}

\noindent
Intuitively, the problem is that there may be no combination of seats that would
give each party its weighted lower quota. As a workaround, we restrict the lower quota of 
party $p$ to the values $p$ can achieve with the number of seats it 
deserves, i.e., $\ell^\#(p) = \lfloor k\cdot\nicefrac{v_{p}}{n}\rfloor$. 
We then use this quantity to determine the party's \emph{obtainable lower quota of weights}
\(
    \ell^o(p) =  
    \max\left\{w\in R(\vec{w})_{[\ell^\#(p)]} \mid w\leqslant \quota(p)\right\}.
\)
%
We can now define our first proportionality property.

\begin{definition}[Obtainable Weighted-seat Lower Quota, $\WLQo$]
A seat assignment $\vec{s}$ provides $\WLQo$ 
if for every party $p$ it is the case that $\rep{p}{\vec{s}} \geqslant \ell^o(p)$.
\end{definition}

\noindent
Note that for unit weights, $\WLQo$ is equivalent to the standard lower quota.
But in the weighted setting, computing $\ell^o(p)$ 
requires solving a \textsc{SubsetSum} problem and can hence not be done in polynomial time,
unless $\P = \NP$.%
\footnote{We assume that weights are encoded in binary.}
Still, there is good news for simple scenarios.

\begin{proposition}\label{prop:WLQ_exists_2_parties}
For every election instance with two parties there exists a seat assignment 
that satisfies $\WLQo$.
\end{proposition}


\noindent
But with more than two parties,
$\WLQo$ is not always achievable. 

\begin{proposition}\label{prop:WLQ_may_not_exist}
There are election instances for which there exists no seat assignment that provides~$\WLQo$.
\end{proposition}

\begin{proof}
Consider vote vector $\prof{v} = (1,1,1)$ for three parties, and weight vector $\vec{w} = (3,2,1)$. 
We get $\ell^o(p) = 2$ for each party $p\in[3]$. But there exists no seat assignment that 
provides at least a weight of~$2$ to all three parties.
\end{proof}

\noindent
While $\WLQo$ cannot always be satisfied, 
one might still ask for a WSAM that delivers an allocation satisfying $\WLQo$ 
on instances where this is possible. 
Unfortunately, the following result shows that 
this is computationally intractable. 
The proof 
involves a reduction from the well-known \NP-complete problem $\textsc{Partition}$.

\begin{proposition}\label{prop:WLQ_is_NP-hard_to_find}
    If there exists a polynomial-time algorithm $\alpha$ that finds a seat assignment $\vec{s}$ that provides $\WLQo$ whenever such a seat assignment exists, then $\P=\NP$.
    This holds even when restricted to the case where there are only two parties.
\end{proposition}


\noindent
But with extra assumptions, we obtain a more positive \new{result}.

\begin{proposition}\label{prop:WLQ_pseudopoly_to_find}
    \new{For} a constant number of parties and weights in $\vec{w}$ \new{that are} polynomial in the input size, finding a seat assignment $\vec{s}$ that provides $\WLQo$ can be done in polynomial time, assuming such a seat assignment exists.
\end{proposition}
\begin{proof}
    We describe a dynamic programming algorithm that finds a seat assignment $\vec{s}$ providing $\WLQo$,whenever one exists. 
    
    Consider an election instance with $m$ parties.
    The algorithm works as follows. For each $i\in [k]$ (where $i$ represents the number of seats assessed thus far), it computes $\mathcal{W}_{i}$ which is a set of tuples of the form $(W_{1},\ldots,W_{m})$. Here, $W_{p}$ indicates the sum of seat weights of the seats assigned to party $p\in[m]$. 

    Each $\mathcal{W}_{i+1}$ can be computed using $\mathcal{W}_{i}$ and the weight $w_{i+1}$, by looking at every combination of some tuple in $\mathcal{W}_{i}$ and some choice of party to assign the weight-$w_{i+1}$ seat to. Once $\mathcal{W}_{k}$ is computed, we can check every tuple in $\mathcal{W}_{k}$ and for each tuple, assess whether it satisfies $\WLQo$ (which can be done in polynomial time). Specifically, this check can be done for each tuple $(W_{1},\ldots,W_{m})$ by assessing whether $W_{p}\geqslant \ell^{o}(p)$ for every party $p\in[m]$. From the assumption on the weights in $\vec{w}$ and the observation that computing $\ell^{o}$ requires solving an instance of Subset Sum, we can apply a dynamic programming algorithm for the latter problem to compute $\ell^{o}(p)$ in polynomial time. And finally, there are at most $\omega^{2m}$ such tuples, which is polynomially many in the input size due to the assumptions on the weights in $\vec{w}$ and the number of parties being constant. Thus, this algorithm runs in polynomial time.
\end{proof}


\noindent
The assumptions of Proposition~\ref{prop:WLQ_pseudopoly_to_find} may be restrictive, 
but they fit the scenarios envisioned for our model, which
are not likely to feature a number of parties, or weight values, exponential in the input size.\footnote{We note that it remains unclear whether there exists a pseudo-polynomial-time algorithm for the case of a superconstant number of parties.}

We have, as of yet, made no inroads towards our goal of finding an achievable lower quota property. 
To do so, it is helpful to
\new{view} lower quota in the unit-weight setting 
\new{from a different perspective:}
instead of 
\new{thinking of it as the}
closest value to the quota that can be obtained in practice, 
we interpret it as guaranteeing that each party~$p$ is at most one seat away from \emph{surpassing} its quota. To make this interpretation 
work with weighted seats one must specify which seat, amongst those not assigned to it, a party has to additionally receive in order to surpass its quota. We parse this in \new{three} ways.

\begin{definition}[WLQ up to one seat, WLQ-1]
    A seat assignment $\vec{s}$ provides WLQ-1 if, 
    for every party $p$, either $\rep{p}{\vec{s}} \geqslant \quota(p)$, or
    there exists some seat $t\in [k]\setminus{\{t' \in \vecPos{p}{\vec{s}}\}}$ 
    such that $\rep{p}{\vec{s}} + w_t > \quota(p)$.
\end{definition}

\begin{definition}[WLQ up to any seat, WLQ-X]
    A seat assignment $\vec{s}$ provides WLQ-X if, for every party $p$, either $\rep{p}{\vec{s}} \geqslant \quota(p)$ or for every seat $t \in [k]\setminus{\{t' \in \vecPos{p}{\vec{s}}\}}$, it holds that $\rep{p}{\vec{s}} + w_t > \quota(p)$.
\end{definition}

    \begin{definition}[WLQ up to any seat from a sufficiently represented party, WLQ-X-r]
        A seat assignment $\vec{s}$ provides WLQ-X-r if, for every party $p$, either $\rep{p}{\vec{s}} \geqslant \quota(p)$ or for every seat $t \in \{t' \in \vecPos{p^{*}}{\vec{s}}\mid p^{*}\in[m]\setminus\{p\},\rep{p^{*}}{\vec{s}} > \quota(p^{*})\}$, it holds that $\rep{p}{\vec{s}} + w_t > \quota(p)$.
    \end{definition}   

\noindent
WLQ-1 states that for each party $p$, there exists a seat it can additionally receive so as to surpass $q(p)$. WLQ-X states that each party $p$ would surpass $q(p)$ if it were to receive any one of the additional seats. \new{WLQ-X-r can then be seen as a weakening of WLQ-X where not all seats are considered, but only the seats that have been assigned to parties that have exceeded their representation quota. The intuition behind this requirement is that if one party receives more than its quota, then this is justified by the fact that we could give none of its seats to another party without that party exceeding its quota.} 
Observe that all three axioms are equivalent to lower quota if restricted to unit-weight instances. 

Let us clarify the relations between these requirements. 
Clearly, WLQ-X implies WLQ-X-r, which in turn implies WLQ-1. It turns out that $\WLQo$ is incomparable to WLQ-X-r (and thus to WLQ-X).

\begin{example}
Consider two parties, votes $\prof{v} = (1,1)$ and
weights $\vec{w} = (97,1,1,1)$. For each party $p\in[2]$,
we have $\quota(p) = 50$ and $\ell^o(p) = 2$. The seat assignment $\vec{s} = (1,1,2,2)$ satisfies $\WLQo$ but not \new{WLQ-X-r since party~$1$ is sufficiently represented and} party~$2$ could also receive seat $2$ without surpassing $\quota(2) = 50$.
\end{example} 

\noindent
In the other direction, Proposition~\ref{prop:WLQ_is_NP-hard_to_find} and (the upcoming) Proposition~\ref{prop:WLQ_X_exists_m=2} give us that WLQ-X does not imply $\WLQo$\new{,} assuming $\P\neq\NP$ 
(see the Appendix for an explicit example showing this).
We follow up by investigating the relationship between $\WLQo$ and WLQ-1.  

\begin{proposition}\label{prop:WLQo_implies_WLQ-1}
    $\WLQo$ implies WLQ-1.
\end{proposition}

\noindent
Are the new axioms easier to satisfy than $\WLQo$?
First, for the two-party case, we find that not only can WLQ-X always be provided, 
but it is even possible to do so efficiently.

\begin{proposition}\label{prop:WLQ_X_exists_m=2}
        For two parties, a seat assignment providing WLQ-X always exists and can be found in polynomial time.
\end{proposition}


\noindent
In other words, for two parties, we find that WLQ-X is easier to satisfy than $\WLQo$. Unfortunately, this does not extend to the case of more than two parties as a result due to \citet{Aziz2020} can be interpreted as showing that a seat assignment providing WLQ-X may not exist for election instances with three or more parties.
While determining how difficult it is to decide whether an assignment providing WLQ-X is possible for a given scenario is left for future work, a minor adjustment to the dynamic programming algorithm of Proposition~\ref{prop:WLQ_pseudopoly_to_find}
yields the following positive result under certain assumptions.
 
\begin{proposition}\label{prop:WLQ-X_pseudopoly_to_find}
    Given a constant number of parties and the weights in $\vec{w}$ being polynomial in the input size, finding a seat assignment $\vec{s}$ that provides WLQ-X can be done in polynomial time, assuming such a seat assignment exists.
\end{proposition}
\begin{proof}
    If we alter the the dynamic programming algorithm from Proposition~\ref{prop:WLQ_pseudopoly_to_find} to also keep track of the smallest seat weight $l_p$ assigned to each party $p$ (alongside its sum of seat weights $W_{p}$) within the tuples in $\mathcal{W}_{i}$, then we can use the modified, final set of tuples $\mathcal{W}_{k}$ to check in polynomial time whether WLQ-X is satisfied. 
\end{proof}
    
\noindent
Following the mostly negative results regarding $\WLQo$ and WLQ-X, 
we take aim at the weaker requirement of WLQ-X-r, 
and (finally!) find a positive result.



\begin{theorem}\label{thm:greedy_satisfies_WLQ-X-r}
\new{The Greedy method satisfies WLQ-X-r.}
\end{theorem}

\begin{proof}
Suppose \new{WLQ-X-r} is violated by some seat assignment $\vec{s}$ returned by the Greedy method. Let party~$p_{x}$ be the party that witnesses it, i.e., $\rep{p_{x}}{\vec{s}} \,\new{<}\, \rep{p_{x}}{\vec{s}} + w_{t} \,\new{\leqslant}\, \quota(p_{x})$ 
\new{for some seat $t \in \{t' \in \vecPos{p^{*}}{\vec{s}}\mid p^{*}\in[m]\setminus\{p\},\rep{p^{*}}{\vec{s}} > \quota(p^{*})\}$.}
As party~$p_{x}$ has less than $\quota(p_{x})$ in representation, there must be a party~$p_{y}$ where $\rep{p_{y}}{\vec{s}} > \quota(p_{y})$. Let $h$ be the round after which party~$p_{y}$ has more than $\quota(p_{y})$ in representation (so party~$p_{y}$ was assigned seat $h$). By choice of round $h$, we have $g_{p_{y}}(h) + w_{h} > \quota(p_{y})$, so it holds that $w_{h} > \quota(p_{y}) - g_{p_{y}}(h)$. Thus, we have that $w_{h} > \quota(p_{y}) - g_{p_{y}}(h)$, and since  party~$p_{x}$ was not assigned seat $h$, we know that $\quota(p_{y}) - g_{p_{y}}(h) \geqslant \quota(p_{x}) - g_{p_{x}}(h)$. It then follows that $\quota(p_{x}) < g_{p_{x}}(h) + w_{h} \leqslant \rep{p_{x}}{\vec{s}}+ w_{h}$. So seat $h$ is enough for party~$p_{x}$ to reach $\quota(p_{x})$, ands the same holds for seats assigned to party~$p_{y}$ in prior rounds (given that seats are assigned in non-increasing order).
\end{proof}

\noindent
Recall that in standard apportionment, LRM is known to satisfy LQ~\cite{BalinskiYoung1982}.
In this light, Theorem~\ref{thm:greedy_satisfies_WLQ-X-r} further justifies 
the Greedy method as a weighted proxy of LRM. 
Recall furthermore that in the standard setting LQ is also satisfied by D'Hondt; 
we now find that D'Hondt's weighted equivalent, $\DHondt$, satisfies WLQ-X-r. 

\begin{theorem}\label{thm:dhondt_satisfies_WLQ-X-r}
    The $\DHondt$ method satisfies \new{WLQ-X-r}.
\end{theorem}

\begin{proof}
        For a seat assignment $\vec{s}$ returned by $\DHondt$, for the sake of contradiction, assume that there is a party $p_{x}$ such that \new{there exists some $t \in \{t' \in \vecPos{p^{*}}{\vec{s}}\mid p^{*}\in[m]\setminus\{p\},\rep{p^{*}}{\vec{s}} > \quota(p^{*})\}$ such that $\rep{p_{x}}{\vec{s}} < \rep{p_{x}}{\vec{s}} + w_{t} \leqslant \quota(p_{x})$.}
        
        Thus, we know that $\nicefrac{v_{p_{x}}}{(g_{p_{x}}(k)+w_{t})} \,\new{\geqslant}\, \nicefrac{v_{p_{x}}}{\quota(p_{x})}$ for some \new{$t \in \{t' \in \vecPos{p^{*}}{\vec{s}}\mid p^{*}\in[m]\setminus\{p\},\rep{p^{*}}{\vec{s}} > \quota(p^{*})\}$}, where $g_{p_{x}}(k)$ is the total weight assigned to party $p_{x}$ at $\DHondt$'s conclusion. This gives us the following:       
        \begin{equation}\label{eq:w-D'Hondt_proof_1}
            \frac{v_{p_{x}}}{g_{p_{x}}(k)+w_{t}} \,\new{\geqslant}\, \frac{v_{p_{x}}}{\quota(p_{x})} = \frac{v_{p_{x}}}{\omega\cdot\nicefrac{v_{p_x}}{n}} = \frac{n}{\omega}
        \end{equation}
        During $\DHondt$, there must be some round $h$ where some party $p_{y}\neq p_{x}\in [m]$ is assigned weight $w_{h}$ such that $\nicefrac{n}{\omega} \,\new{>}\, \nicefrac{v_{p_{y}}}{(g_{p_{y}}(h)+w_{h})}$. Assume otherwise and that for every party $p\in [m]\new{\setminus\{p_{x}\}}$ it holds that $\nicefrac{v_{p}}{g_{p}(k)} \,\new{\geqslant}\, \nicefrac{n}{\omega}$ after $\DHondt$'s $k$ rounds. Then we have that $\omega\cdot\nicefrac{v_{p}}{n} \,\new{\geqslant}\, g_{p}(k)$ for all $p\in [m]\new{\setminus\{p_{x}\}}$. Summing over all parties \new{with $\omega\cdot\nicefrac{v_{p_{x}}}{n} > g_{p_{x}}(k)$ for party~$p_{x}$}, we get $\sum_{p\in [m]} \omega\cdot\nicefrac{v_{p}}{n} = \omega > \new{g_{p_{x}}(k) + }\sum_{p\in [m]\new{\setminus\{p_{x}\}}} g_{p}(k)$, so $\DHondt$ did not assign all of the weight, contradicting its definition. So, there must exist some round $h$ where for some party $p_{y}$, we have:
        \begin{equation}\label{eq:w-D'Hondt_proof_2}
           \frac{n}{\omega} \,\new{>}\, \frac{v_{p_{y}}}{g_{p_{y}}(h)+w_{h}}
        \end{equation}
        Since weight $w_{h}$ was assigned to party $p_{y}$ in round $h$, and not to party $p_{x}$, we have that $\nicefrac{v_{p_{y}}}{(g_{p_{y}}(h)+w_{h})} \geqslant \nicefrac{v_{p_{x}}}{(g_{p_{x}}(h)+w_{h})}$, where $h\in \{t' \in \vecPos{p^{*}}{\vec{s}}\mid p^{*}\in[m]\setminus\{p\},\rep{p^{*}}{\vec{s}} > \quota(p^{*})\}$. And also considering the fact that $g_{p_{x}}(h)\leqslant g_{p_{x}}(k)$, it follows that: 
        \begin{equation}\label{eq:w-D'Hondt_proof_3}
            \frac{v_{p_{y}}}{g_{p_{y}}(h)+w_{h}} \geqslant \frac{v_{p_{x}}}{g_{p_{x}}(h)+w_{h}} \geqslant \frac{v_{p_{x}}}{g_{p_{x}}(k)+w_{h}}
        \end{equation}
        Putting equations (\ref{eq:w-D'Hondt_proof_1}), (\ref{eq:w-D'Hondt_proof_2}), and (\ref{eq:w-D'Hondt_proof_3}) together, it follows that $\nicefrac{n}{\omega} \,\new{>}\, \nicefrac{v_{p_{y}}}{(g_{p_{y}}(h)+w_{h})} \,\new{\geqslant}\, \nicefrac{n}{\omega}$.
        %
        This is a contradiction, so no such party $p_{x}$ can exist. \new{Note that we considered a seat weight $w_{h}$ assigned to some party $p_{y}$ in round $h$, such that $p_{y}$ surpasses its quota. And such a weight $w_{h}$ is sufficient in aiding party~$p_{x}$ in reaching $\quota(p_{x})$. This holds for all seats assigned to party~$p_{y}$ before round $h$ (as such seats $h^{*}$ have weight $w_{h^{*}}\geqslant w_{h}$), and also those seats assigned to party~$p_{y}$ after round $h$ (as such seats $h^{*}$ are only assigned to party~$p_{y}$, and not some party~$p_{x}$ below its quota $\quota(p_{x})$ in that round, if the weight $w_{h^{*}}$ would lead to party~$p_{x}$ reaching said quota).}         
\end{proof}

\noindent
This improves on a result of \citet{ChakrabortySS21} stating that D'Hondt satisfies an axiom called WPROP1 (see their Theorem~4.9),\footnote{WPROP1 is similar to WLQ-1 but defined with a weak inequality in the condition on the existence of a seat of sufficient weight. This axiom has been frequently studied in settings more general than ours \cite{Aziz2020,LiLW2022,HVN2024,AzizLMWZ2024}, along with generalisations such as that of \citet{ChakrabortySS2024revisit}.} which is weaker than WLQ-X-r. The importance of our findings is strengthened by observing that the `up to any' properties are, in many scenarios, much stronger than the equivalent `up to one' properties, especially if the values of objects vary a lot. For instance, consider how our WSAMs would handle the allocation of non-liquid assets in a bankruptcy. Using the monetary value of the assets as their weights, we might have a few very valuable assets (e.g., a house or other property), together with assets of much lower value (e.g., furniture). In such a case `up to one' properties can become essentially meaningless, while `up to any' properties are still meaningful. 
Crucially, our stronger result does not only stem from our restricted setting but also from our use of weighted $\DHondt$, as standard D'Hondt---used by \citet{ChakrabortySS21}---does not satisfy WLQ-X-r in our setting. 
The next example shows this.

\begin{namedexample}[Standard D'Hondt fails WLQ-X-r]\label{exm:standardDHondt_fails_WLQ-x-r}
Consider two parties with votes $\prof{v} = (10,2)$ and the weight vector $\vec{w} = (10,1,1)$. Standard D'Hondt assigns all three seats to party~$1$: in the three rounds party~$1$ has the ratios $10$, $5$, and $2.5$, respectively, versus party~$2$'s ratio of $2$ in all three rounds. Thus, party~$2$ has representation of $0$ from the resulting seat assignment and no weight-$1$ seats are enough to add so that party~$2$ exceeds its quota of $\quota(2) = 2$. However, observe that the seat assignment determined by standard D'Hondt provides WLQ-1, while our WSAM $\DHondt$ returns the seat assignment $\vec{s} = (1,2,2)$. Also, this seat assignment $\vec{s}$ not only provides WLQ-X-r, but it is intuitively a much fairer outcome. 
\end{namedexample}

\noindent
Example \ref{exm:standardDHondt_fails_WLQ-x-r} illustrates the usefulness of objective weights for seats,
and motivates our use of WSAMs to satisfy the stronger `up-to-any' properties.
%
%
On the other hand, $\Adams$ fails even WLQ-1 (an example can be found in the Appendix).
This is unsurprising, as it is known to violate LQ, 
which, as mentioned above, is equivalent to WLQ-1 in the standard unit-weight case \cite{BalinskiYoung1982}.


\section{Upper Quota}\label{sec:upper-quota}

\new{Parties should get at least as many seats as they deserve (lower quota), but also not more than appropriate:
the latter bound is captured by an \emph{upper quota} (UQ) property.}
In standard apportionment,
UQ states that a party~$p$ amassing $v_{p}$ of the $n$ votes should receive
\emph{at most} $\lceil k\cdot\nicefrac{v_{p}}{n} \rceil$ of the $k$ seats \cite{BalinskiYoung1982}.
As with lower quota, 
there is no hope of
satisfying the \new{na\"{\i}ve} weighted upper quota notion defined by
$\lceil \omega\cdot\nicefrac{v_{p}}{n} \rceil$. 
We \new{can then} define an obtainable upper quota 
\new{as for}
the obtainable weighted lower quota $\ell^o(p)$, \new{but}
the results for the corresponding axiom are very similar to the results for $\WLQo$
and, as we show in the Appendix, they are similarly negative.
%
%
We thus move on to `up-to-one/any' notions, which will allowe us to define satisfiable LQ axioms.

\begin{definition}[WUQ up to one seat, WUQ-1]
A seat assignment $\vec{s}$ provides WUQ-1 if, for every party $p$, either $\rep{p}{\vec{s}} \leqslant \quota(p)$ or there exists some seat $t\in \vecPos{p}{\vec{s}}$ such that $\rep{p}{\vec{s}} - w_{t} < \quota(p)$.
\end{definition}

\begin{definition}[WUQ up to any seat, WUQ-X]
A seat assignment $\vec{s}$ provides WUQ-X if, for every party $p$, either $\rep{p}{\vec{s}} \leqslant \quota(p)$ or for every seat $t\in \vecPos{p}{\vec{s}}$, it holds that $\rep{p}{\vec{s}} - w_{t} <  \quota(p)$.
\end{definition}

\noindent
WUQ-X states that, for every party $p$, disregarding any seat it received would take it below $q(p)$, while for WUQ-1 there need only exist one such seat assigned to party~$p$ to take it below $q(p)$. With unit weights, both axioms reduce to UQ. Observe that there is no natural way of defining a counterpart to WLQ-X-r
. Finally, note that WUQ-X 
implies WUQ-1.

%
%
%

Next, we ask whether upper-quota axioms can be satisfied. A natural candidate
is $\Adams$, as it is known to satisfy UQ
for unit weights~\cite{BalinskiYoung1982}. 
Indeed, $\Adams$ even satisfies the stronger notion of WUQ-X\new{,} in stark contrast to WLQ-X, which is not satisfiable in general. 
We show that $\Adams$ satisfies WUQ-X by showing that it satisfies the following \emph{envy-freeness} axiom 
\cite{springerHY24}.

\begin{newenv}
    \begin{definition}[Weighted envy-freeness up to any seat, WEFX]
    A seat assignment $\vec{s}$ provides WEFX if for any two parties $p_{x},p_{y}$,
    it holds for every seat $t\in \vecPos{p_{y}}{\vec{s}}$ that $\nicefrac{\rep{p_{x}}{\vec{s}}}{v_{p_{x}}} \geqslant \nicefrac{(\rep{p_{y}}{\vec{s}} - w_{t})}{v_{p_{y}}}$.
\end{definition}
\end{newenv}

\noindent
WEFX ensures that no party prefers the representation afforded to another party. Conceptually, this is similar to upper quota, stating that no party is represented more than it deserves.
The next result provides a formal connection between envy-freeness and upper quota.

\begin{newenv}
    \begin{proposition}\label{Prop:WEFX-WUQ-X}
        WEFX implies WUQ-X.
    \end{proposition}

\end{newenv}

\noindent
On the other hand, it is easys to see that 
the following axiom, WEF1, does not imply WUQ-X (a counterexample is in the Appendix).\footnote{One can also show that the weakening of WEF1 known as WWEF1 \cite{chakrabortyISZ21envy} does not imply WUQ-1, see the Appendix for an example.}

\begin{definition}[Weighted envy-freeness up to one seat, WEF1]
    A seat assignment $\vec{s}$ provides WEF1 if for any two parties $p_{x},p_{y}$,
    there exists some seat $t\in \vecPos{p_{y}}{\vec{s}}$ such that $\nicefrac{\rep{p_{x}}{\vec{s}}}{v_{p_{x}}} \geqslant \nicefrac{(\rep{p_{y}}{\vec{s}} - w_{t})}{v_{p_{y}}}$.
\end{definition}

\noindent
That WEF1 implies WUQ-1 follows from similar reasoning to that showing that WEFX implies WUQ-X.  

\noindent

\begin{newenv}
    \begin{theorem}
        The $\Adams$ method satisfies WEFX.
    \end{theorem}

    \begin{proof}
        Suppose there are two parties $p_{x},p_{y}$ with $\nicefrac{\rep{p_{y}}{\vec{s}}}{v_{p_{y}}} < \nicefrac{\rep{p_{x}}{\vec{s}}}{v_{p_{x}}}$, i.e., party~$p_{y}$ envies party~$p_{x}$. Now, consider the last seat $t$ that was assigned to party~$p_{x}$ by $\Adams$ in round $h$. Since this seat was assigned to party~$p_{x}$, we have that
        \(
            \nicefrac{v_{p_{x}}}{g_{p_{x}}(h)} \geqslant \nicefrac{v_{p_{y}}}{g_{p_{y}}(h)}.
        \)
        And as seat $t$ was the last seat assigned to party~$p_{x}$, we get:
        \(
            \nicefrac{\rep{p_{y}}{\vec{s}}}{v_{p_{y}}} \geqslant 
            \nicefrac{g_{p_{y}}(h)}{v_{p_{y}}} \geqslant 
            \nicefrac{g_{p_{x}}(h)}{v_{p_{x}}} = \nicefrac{(\rep{p_{x}}{\vec{s}} - w_{t})}{v_{p_{x}}}.
        \)
        So, removing seat $t$ from party~$p_{x}$ means party~$p_{y}$ no longer envies $p_{x}$, and as all seats assigned to $p_{x}$ prior to seat $t$ have weight at least as large as $w_t$, removing any of these seats is sufficient to remove $p_{y}$'s envy.
    \end{proof}
\end{newenv}

\noindent
It is known that WEFX can be achieved in our setting due to this being the case for the setting of weighted fair division under certain restrictions imposed for that model \cite{springerHY24}. Still, WEFX being achievable with the use of a simple method such as our weighted $\Adams$ is an important new insight that strengthens the practical relevance of this positive finding. This improves on results showing that Adams can be used to achieve WEF1 \cite{chakrabortyISZ21envy}. Indeed, finding a rule that satisfies WEFX in the more general setting of \citet{ChakrabortySS21} has recently been shown to be impossible by \citet{springerHY24}, while showing the existence of EFX allocations in the standard fair division setting remains one of the major open questions in fair division~\cite{AmanatidisBFV22}.

Our next result is a direct corollary of Proposition~\ref{Prop:WEFX-WUQ-X}.

\begin{newenv}
\begin{corollary}
\new{The} \notnew{$\Adams$} \new{method} \notnew{satisfies WUQ-X.}
\end{corollary}
\end{newenv}

\noindent
In view of the fact that $\Adams$ does not satisfy WLQ-1, \new{can} envy-freeness and lower quota be satisfied at the same time\new{?}
\new{This is not possible, as WEF1 and WLQ-1 are incompatible.} \citet{chakrabortyISZ21envy} showed this for the case where voters may value the seats differently, but their work also shows that this negative result holds in our more restricted setting with identical valuations.  

\begin{newenv}
\begin{proposition}
WEF1 and WLQ-1 are incompatible.
\end{proposition}
\end{newenv}

\noindent
\new{Can we} at least satisfy upper- and lower-quota axioms at
the same time\new{?} $\DHondt$ does not satisfy UQ in the unit-weight case so it cannot
satisfy WUQ-1 (see the Appendix for an example). The Greedy method, however, is a contender
due to its connection to LRM, which satisfies UQ \cite{BalinskiYoung1982}.
\new{As it satisfies WLQ-1 it cannot satisfy WEF1, but it does satisfy WUQ-X.}

\begin{theorem}\label{thm:Greedy_WUQ-X}
The Greedy method satisfies WUQ-X.
\end{theorem}

\begin{proof}
Take a seat assignment $\vec{s}$ constructed by the Greedy method. Assume there is a party~$p_{x}$ that
received more representation than $\quota(p_{x})$, i.e., $\rep{p_{x}}{\vec{s}} = g_{p_{x}}(k) > \quota(p_{x})$, otherwise, WUQ-X is satisfied. Let $t$ be the round after which
$g_{p_{x}}(t) > \quota(p_{x})$ holds, and so $g_{p_{x}}(t) - w_{t} < \quota(p_{x})$ also holds. 
We argue that party~$p_{x}$ does not get assigned any seat $t^* > t$. Observe that in every round $t'\in[k]$, there always exists a party~$p_{y}$ such that $\quota(p_{y}) - g_{p_{y}}(t')\geqslant 0$ as we have $\sum_{p \in [m]} \quota(p) = \omega$. It then follows, for every round $t^* > t$, that $\quota(p_{y}) - g_{p_{y}}(t^*) \geqslant 0 > \quota(p_{x}) - g_{p_{x}}(t^*)$ and hence, party~$p_{x}$ cannot be assigned seat $t^*$. \new{Thus, the Greedy method satisfies {WUQ-1}. However, for all seats $j < t$ assigned to party $p_{x}$ thus far, since $w_j \geqslant w_{t}$, it follows that $g_{p_{x}}(t) - w_j < \quota(p_{x})$, and this means that removing any seat assigned to party $p$ suffices for this party to not be above its weighted quota and thus, WUQ-X is also satisfied.}
\end{proof}


\section{House Monotonicity}\label{sec:house_mon}

Why focus on $\DHondt$ or $\Adams$ when the Greedy method satisfies both WLQ-1 and WUQ-X?
The answer lies with \emph{house monotonicity},
which asks that an increase in the number of available seats should not  
make any party worse off. Historically, its failure has 
been the cause of much political animosity \cite{Szpiro2010-pz}.
In standard apportionment\new{,} LRM satisfies LQ and UQ but fails house monotonicity,
whereas divisor methods satisfy it \cite{BalinskiYoung1982,Pukelsheim2014}.
We note that \citet{ChakrabortySS21} also study house monotonicity (under the name of \emph{resource monotonicity}) in their setting of weighted fair division.
We first consider a strong generalisation of the axiom of house monotonicity that states that the addition of a seat of \emph{any weight} does not lead to a decrease of any party's weight representation.


\begin{definition}[Full House Monotonicity, full-HM]
    We say a WSAM $F$ satisfies full-HM if for every election instance $(\prof{v},\vec{w})$ and every $w^{*}\in \mathbb{N}_{\geqslant 1}$ such that $\vec{w}^{*} = (w')_{w'\in W}$ is a non-increasing weight vector where $W = \{w\in \vec{w}\}\cup\{w^{*}\}$, it holds for $\vec{s}\in F(\prof{v},\vec{w})$ and $\vec{s}^{*}\in F(\prof{v},\vec{w}^{*})$ that $\rep{p}{\vec{s}^{*}} \geqslant \rep{p}{\vec{s}}$ for every party $p\in[m]$. 
\end{definition}


\noindent
Full-HM can always be satisfied, e.g., by non-weighted divisor methods \cite{ChakrabortySS21}.
But, as we saw before, these methods do not satisfy the strongest quota axioms, WLQ-X-r and WUQ-X. 
Unfortunately, the more desirable WSAMs we defined do not satisfy full-HM.

\begin{proposition}\label{prop:Adams_DHondt_fail_fullhm}
    $\Adams$ and $\DHondt$ fail full-HM. 
\end{proposition}
\begin{proof}
    Let us first consider $\Adams$ and consider an instance with three parties having votes $\prof{v} = (5,5,2)$ and a weight vector $\vec{w} = (8,8,3,2)$. The seat assignment returned by $\Adams$ is $\vec{s} = (1,2,3,3)$, assuming that ties are broken according to the ordering of $\prof{v}$. Thus, we have $\rep{3}{\vec{s}} = 5$ for party $3$. Suppose a weight-$4$ seat is added to $\vec{w}$ so as to obtain the weight vector $\vec{w}^{*} = (8,8,4,3,2)$. Then $\Adams$ returns $\vec{s}^{*} = (1,2,3,1,2)$. So party~$3$'s representation changes from $\rep{3}{\vec{s}} = 5$ to $\rep{3}{\vec{s}^{*}} = 4$.
    
    Consider an instance with three parties having votes $\prof{v} = (21,10,10)$ and a weight vector $\vec{w} = (2,2)$. The seat assignment returned by $\DHondt$ is $\vec{s} = (1,1)$, giving both seats to party $1$ who obtain $4$ in representation. Suppose a weight-$3$ seat is added to $\vec{w}$ so as to obtain $\vec{w}^{*} = (3,2,2)$. Then $\DHondt$ returns $\vec{s}^{*} = (1,2,3)$, assigning a seat to each party. So party~$1$'s representation changes from $\rep{1}{\vec{s}} = 4$ to $\rep{1}{\vec{s}^{*}} = 3$.
\end{proof}    

\noindent
We leave open whether there is a WSAM that satisfies full-HM along with some of our stronger proportionality axioms.
Instead, we try to achieve positive results by restricting the weight associated with an election's additional seat. This consideration leads us to the following, weaker axiom---where, as opposed to full-HM, the extra $k+1$-st seat must have a weight no larger than any of the original $k$ seats.

\begin{definition}[Minimal House Monotonicity, min-HM]
We say a WSAM $F$ satisfies min-HM if for every election instance $(\prof{v},\vec{w})$ and every $w^{*}\in \mathbb{N}_{\geqslant 1}$ such that  $w^{*}\leqslant w_{k}$ and $\vec{w}^{*} = (w')_{w'\in W}$ is a non-increasing weight vector where $W = \{w\in \vec{w}\}\cup\{w^{*}\}$, it holds for $\vec{s}\in F(\prof{v},\vec{w})$ and $\vec{s}^{*}\in F(\prof{v},\vec{w}^{*})$ that $\rep{p}{\vec{s}^{*}} \geqslant \rep{p}{\vec{s}}$ for every party $p\in[m]$. 
\end{definition}

\noindent
Positively, divisor methods clearly satisfy min-HM, as all seats prior to an extra seat are assigned in the same way.

\begin{proposition}
    All divisor methods satisfy min-HM.
\end{proposition}

\noindent
While much weaker than full-HM, min-HM is enough to differentiate between our WSAMs: the Greedy method---somewhat unsurprisingly given LRM's failure of HM with unit weights---fails min-HM. 

\begin{proposition}\label{prop:Greeedy_fails_minHM}
    The Greedy method fails min-HM. 
\end{proposition}
\begin{proof}
    Consider three parties, the vote vector $\prof{v} = (5,4,1)$, and a weight vector $\vec{w} = (4,3,2)$. The seat assignment returned by the Greedy method is $\vec{s} = (1,2,3)$. Now for the weight vector $\vec{w}^{*} = (4,3,2,1)$, the method assigns the first two seats to parties $1$ and $2$, respectively. In round $3$, note that $\omega\cdot\nicefrac{v_{p}}{n} - g_{p}(3) = 1$ for parties $p\in[3]$. Suppose that party~$1$ is assigned the third seat via tie-breaking. For the next round, parties $2$ and $3$ remain equally entitled to last seat. Suppose that tie-breaking leads to this seat being assigned to party~$2$. The method then returns the seat assignment $\vec{s}^{*} = (1,2,1,2)$ with party~$3$ receiving less representation than in the original election instance.
\end{proof}

\noindent
We leave  
\new{the task of investigating} other natural weakenings of full-HM 
to future work. 


\begin{table}[t]\small
\centering
\caption{\label{tab:axiomTable}
Summary 
whether a WSAM satisfies (\cmark) or violates (\xmark) a given axiom
 (for axioms satisfied by at least one of our WSAMs).\vspace*{15pt}}
\begin{tabular}{@{\ }l@{\ }c@{\ \ \ }c@{\ \ \ }c@{\ \ \ }c@{\ \ \ }c@{\ }} \toprule
                   &  \new{WEFX} & \new{WLQ-X-r}  &   WUQ-X & WUQ-1 & min-HM \\ \midrule
    $\Adams$       & \new{\cmark} & \new{\xmark}  & \cmark & \cmark & \cmark\\
    $\DHondt$      & \new{\xmark} & \new{\cmark}  & \xmark & \xmark & \cmark\\\midrule
    Greedy Method  & \new{\xmark} & \new{\cmark} & \cmark & \cmark & \xmark\\ \bottomrule
\end{tabular}
\end{table}

\section{Bundestag Case Study}\label{sec:bundestag_study}

We now present a case study regarding the allocation of chair positions 
to parties in Bundestag committees.
These are committees with specific responsibilities (e.g., Budget or Defence),
established anew in every political cycle. 
Usually, which party gets to nominate the chair of a committee is the result of negotiation---but when no consensus can be found, which happened in 8 out of 20 parliamentary sessions \new{since 1949}, standard apportionment methods are used.
Crucially, different committees have different size and influence,
so positions differ in value: 
the role of chair of the Budget Committee will be valued more highly than 
that of the Tourism Committee. 

Our objective is to compare the results produced by our weighted apportionment methods 
with the historical results, full details of which are publicly available
\cite{Feldkamp2023,wikipedia-bundestagsausschuss}.
When applying our methods, we interpret the size of a committee as a proxy for its importance. 
Of course, this can only ever be an approximation of the true value of a committee, but we believe it nonetheless can provide a first impression of how the WSAMs we study perform on real-world data.

The existing data covers all 20 legislative periods in Germany between 1949 and 2021. 
For each of these periods, between 4 and 7 parties entered parliament, between 19 and 28 committees were formed, and each committee had between 9 and 49 members.\footnote{In case any relevant data points (such as the size of a committee) changed over the course of a legislative period, we always used the start of that period as our point of reference.} To construct an election instance for a given legislative period, we take the members of parliament to be the voters,
we take the chair positions for the committees of that period to be the seats to be filled, and we use the sizes of those committees as the weights of the seats.

For each of the 20 election instances thus created, we are interested in how the historical Bundestag seat assignment fares in terms of representing parties proportionally and how that assignment compares to the assignments returned by $\Adams$, $\DHondt$, and Greedy.
%
Given a seat assignment~$\vec{s}$, we first ask which of our \new{nine} 
axioms it satisfies. For testing $\WLQo$ and $\WUQo$, we encoded the computations of the obtainable weighted quotas $\ell^{o}(p)$ and $u^{o}(p)$ into Integer Linear Programs (ILP) and employed an ILP solver to compute them efficiently.
As the binary measure of axiom satisfaction 
\new{provides only} limited insight, we introduce 
a \new{finer-grained} measure,
\emph{average distance to the weighted quota}:
    \(
        \delta(\vec{s}) = \nicefrac{1}{n}\cdot\sum\nolimits_{p\in[m]} |\rep{p}{\vec{s}} - \quota(p)|.
    \)

    This leads to eleven measures for the proportionality of a given seat assignment,
    which we use to evaluate the Bundestag allocation.
    The results are summarised in Table~\ref{tab:bundestagTable}.\footnote{The code used to generate these results is publicly available~\cite{code}.} 
    The historical seat assignments perform reasonably well in terms of our measures of quality,
    but both $\DHondt$ and Greedy do markedly better. 
    This is borne out by the rate at which the axioms are satisfied 
    and the results for the distance measure. 
    Notably, not only do the Bundestag seat assignments yield a worse median and maximum distance 
    than all the WSAMs, but $\DHondt$ and Greedy significantly outperform the Bundestag assignments 
    in this metric (significance level $\alpha$ = 0.05). 
    Of the latter group, Greedy 
    produces lower distance results across the board, even compared to $\DHondt$.
    However, these differences between Greedy and $\DHondt$ are not statistically significant.
    $\Adams$ performs badly with respect to the distance measure, but less so with respect to the envy-freeness axioms. 
    
    Despite the small sample size, our results suggest that $\DHondt$ and Greedy are deserving of
    further investigation.


\begin{table}[t]\small
\centering
\caption{\label{tab:bundestagTable}
Summary of results for the 20 Bundestag committee election instances for 1949--2021. 
For each of the four seat assignments, the table shows: $(i)$~for each axiom, the percentage of election instances for which the axiom is satisfied; and $(ii)$~for our distance measure, the median and maximum distances across the 20 election instances.\vspace*{15pt}}
\begin{tabular}{@{\ }l@{\ }c@{\ \ \ }c@{\ \ \ }c@{\ \ \ }c@{\ }} \toprule
       &  \textit{Bundestag}  &  $\Adams$ & $\DHondt$ & Greedy \\
\midrule
$\WLQo$ (\%)   &  $0$   &  $0$    &  $25$     &  $30$ \\
WLQ-X (\%) &  $20$  &  $5$    &  $95$     &  $100$ \\
\new{WLQ-X-r (\%)} &  \new{$20$}  &  \new{$35$}    &  \new{$100$}     &  \new{$100$} \\
WLQ-1 (\%) &  $45$  &  $80$   &  $100$    &  $100$ \\
\midrule
$\WUQo$ (\%)   &  $0$   &  $0$    &  $5$      & $5$ \\
WUQ-X (\%) &  $20$  &  $100$  &  $90$     & $100$ \\
WUQ-1 (\%) &  $50$  &  $100$  &  $100$    & $100$ \\
\midrule
\new{WEFX (\%)} &  \new{$0$}  &  \new{$100$}    &  \new{$10$}     &  \new{$25$} \\
\new{WEF1 (\%)} &  \new{$10$}  &  \new{$100$}    &  \new{$30$}     &  \new{$50$} \\
\midrule
Median $\delta$  &  $17.1$   &  $16.6$       &  $4.7$    &  $2.5$ \\
Maximum $\delta$ &  $66$     &  $30.8$       &  $6.6$    &  $5.9$ \\
\bottomrule
\end{tabular}
\end{table}

\section{Conclusion and Future Work}
\label{sec:conclusion}

We introduced a model for apportionment with weighted seats, 
and generalised standard methods and 
axioms from the apportionment literature to this model. 
While direct generalisations of standard axioms yield (mostly) negative results, 
mild relaxations are amenable to positive results (see Table~\ref{tab:axiomTable}). 
This positive outlook is further supported, particularly for the $\DHondt$ and Greedy methods, by our experimental case study on Bundestag committee assignments.

Besides the open questions regarding house monotonicity (see Section~\ref{sec:house_mon}), 
there are two natural directions for future work. The first is a study of other prominent apportionment properties and rules, notably population monotonicity and the \new{Sainte-Lagu\"{e}} method. 
The second is an extension of the weighted-seat notion to more general settings, e.g., multi-winner voting. 
We recently have taken initial steps in this direction~\cite{ChingomaPhD2025}.
But obtaining positive results in this more general setting is challenging. 
For instance, we find that lifting even the seemingly mild assumption of full supply 
results in the weakest of our axioms, i.e., WLQ-1 and WUQ-1, ceasing to be satisfiable.
Concrete examples are provided in the Appendix. 



\begin{ack}
We are grateful to multiple anonymous reviewers for their helpful feedback.
This work was supported by the Dutch Research Council (NWO) as part of the \emph{Collective Information} project funded under the VICI scheme (grant number 639.023.811); by the Austrian Science Fund (FWF), Austria (grant numbers J4581, 10.55776/PAT7221724, and 10.55776/COE12); by \emph{netidee} F\"orderungen, Austria (\url{https://www.netidee.at}); and by the Vienna Science and Technology Fund (WWTF), Austria (grant number 10.47379/ICT23025).
\end{ack}


\bibliography{ecai25}

\clearpage

\section{Appendix: Supplementary Material}\label{sec:appendix}

In this technical appendix we present the remaining proofs omitted from the body of the paper and provide additional examples to illustrate some of the technical findings discussed in the paper. We also discuss an additional upper-quotas axiom in some detail. The code used for the Bundestag case study is available online~\cite{code}.

\subsection{Supplementary Material for Section~\ref{sec:lower-quota}}

We begin with the proof of Proposition~\ref{prop:WLQ_exists_2_parties} showing that $\WLQo$ can be satisfied when there are only two parties.

\begin{proof}[Proof of Proposition~\ref{prop:WLQ_exists_2_parties}]
Consider an election instance with parties $1$ and $2$.
By definition of the obtainable weighted lower quota, there exists a set \(T\) of seats such that 
$\sum_{t\in T} w(t) = \ell^o(1)$. Let $\vec{s}$ now be a seat assignment such that party~$1$ is assigned all seats in $T$ and party~$2$ gets all of the other seats. By definition,
$\rep{1}{\vec{s}} = \ell^o(1)$ holds for party~$1$. For party~$2$ we have: 
\begin{align*}
    \rep{2}{\vec{s}} & = \omega - \rep{1}{\vec{s}} = \omega - \ell^o(1)\\
                       & \geqslant \omega - \omega\cdot\frac{v_{1}}{n} =  \omega\cdot\frac{n - v_{1}}{n} =
                         \omega\cdot\frac{v_{2}}{n} \geqslant \ell^o(2).
\end{align*}
Hence, $\WLQo$ is satisfied.
\end{proof}

\noindent
Next, we present the proof of Proposition~\ref{prop:WLQ_is_NP-hard_to_find}, establishing the intractability of finding seat assignments that satisfy $\WLQo$. Recall that the claim is that, if there exists an algorithm~$\alpha$ to solve this problem in polynomial time, then $\P=\NP$.

\begin{proof}[Proof of Proposition~\ref{prop:WLQ_is_NP-hard_to_find}]
Assume such an algorithm $\alpha$ exists. 
\new{Recall that the \textsc{Partition} problem asks, for a given multiset $X = \{x_{1},\ldots,x_{k}\}$ of $k$ positive integers, whether there exists a partition of $X$ into two subsets $X_{1}$ and $X_{2}$ such that $\sum_{x\in X_{1}} x = \sum_{x\in X_{2}} x$.
Let $X$ be such an instance of the $\textsc{Partition}$ problem.}
We create an election instance $(\prof{v},\vec{w})$, as follows. 
Set the weight vector $\vec{w} = (x)_{x\in X}$ to be the non-increasing vector of the $k$ elements in $X$. Thus, we have $\omega = \sum_{x\in X}x$. Take two parties with $\prof{v} = (1,1)$, so each party $p\in [2]$ receives half of the total votes. Now let $\vec{s}$ be the seat assignment produced by $\alpha$ on input $(\prof{v},\vec{w})$.
We now prove that $X$ is a positive instance of $\textsc{Partition}$ if and only if 
$\rep{p}{\vec{s}} = \nicefrac{\omega}{2}$ for every $p\in [2]$. 

$(\Rightarrow)$ Assume that $X$  is a positive instance of $\textsc{Partition}$. 
Thus, there exist subsets $X_{1}$ and $X_{2}$ such that $\sum_{x\in X_{1}}x = \sum_{x\in X_{2}}x$. In particular, this means that $\sum_{x\in X_{t}}x = \nicefrac{\omega}{2}$, 
for each $t\in[2]$. Consider the constructed election instance $(\prof{v},\vec{w})$. Each party deserves $\nicefrac{k}{2}$ seats and has a weighted lower quota of $\nicefrac{\omega}{2}$. As $\min\left(|X_1|,|X_2|\right) \leqslant \nicefrac{k}{2}$, there exists a way to receive weight $\nicefrac{\omega}{2}$ with $\nicefrac{k}{2}$ seats. It follows that both parties have an obtainable lower quota of $\nicefrac{\omega}{2}$. Finally, $\WLQo$ is satisfiable, as the seat assignment that gives all seats that correspond to an element of $X_1$ to party $1$ and every seat that corresponds to an element of $X_2$ to party $2$, satisfies $\WLQo$. It follows that in the seat assignment produced by $\alpha$ each party must have representation $\nicefrac{\omega}{2}$.

$(\Leftarrow)$ Assume that $\rep{p}{\vec{s}} = \nicefrac{\omega}{2}$ for both parties $p$.
Let $X_1$ be the set of all elements that correspond to a seat that is allocated to party $1$
and let $X_2$ be the set of all elements that correspond to a seat that is allocated to party $2$.
Then we must have $\sum_{x\in X_{1}}x = \nicefrac{\omega}{2} = \sum_{x\in X_{2}}x$,
and hence $X$ is a positive instance of $\textsc{Partition}$.

To summarise, we can solve $\textsc{Partition}$ in polynomial time by transforming it 
into the election instance $(\prof{v},\vec{w})$, running $\alpha$ on that instance, and finally checking
whether $\rep{p}{\vec{s}} = \nicefrac{\omega}{2}$ holds for both parties $p\in[2]$. As $\textsc{Partition}$ is known to be \NP-complete,
this implies that $\P=\NP$.
\end{proof}

\noindent
We move on to the proof of Proposition~\ref{prop:WLQo_implies_WLQ-1}, claiming that $\WLQo$ implies WLQ-1.

\begin{proof}[Proof of Proposition \ref{prop:WLQo_implies_WLQ-1}]
    Consider some party~$p$ with $\ell^\#(p) = \lfloor k\cdot\nicefrac{v_{p}}{n}\rfloor$. Now suppose we iterate through the seats in non-increasing order of weight and assign $h$ seats to party~$p$ such that $h\leqslant\ell^\#(p)$ and $\sum_{t\in [h]} w_{t} \leqslant \quota(p) < \sum_{t\in [h+1]} w_{t}$. We know that $\sum_{t\in [h]} w_{t} \leqslant \quota(p)$ and that it is obtainable with at most $\ell^\#(p)$ seats, i.e., $\sum_{t\in [h]} w_{t}\in R(\vec{w})_{[\ell^\#(p)]}$. Since $\ell^{o}(p)$ is the maximal value amongst such weights, we get $\ell^{o}(p) \geqslant \sum_{t\in [h]} w_{t}$. Now assume a seat assignment $\vec{s}$ provides $\WLQo$, so we have $\rep{p}{\vec{s}} \geqslant \ell^{o}(p)$ for party~$p$. Assume that $\rep{p}{\vec{s}} < \quota(p)$ and consider the seats assigned to party~$p$ in $\vec{s}$. If they are exactly the same $h$ seats as selected above, then seat $h+1$ works so that $\rep{p}{\vec{s}} + w_{h+1} = \ell^{o}(p) + w_{h+1} \geqslant \sum_{t\in [h]} w_{t} + w_{h+1} > \quota(p)$. Otherwise, party~$p$ did not receive one of those $h$ seats. Now, since we have $w_{t}\geqslant w_{h+1}$ for all those seats $t\in [h]$, it must be the case that $\rep{p}{\vec{s}} + w_{t} > \quota(p)$ holds for every seat $t\in [h]$. Thus, WLQ-1 is satisfied.
\end{proof}

\noindent
Here is an example that shows that WLQ-X does not imply $\WLQo$.

\begin{example}\label{exm:WLQ-X_does_not_imply_WLQo}
Consider three parties, a vote vector $\prof{v} = (3,2,1)$, and a weight vector $\vec{w} = (3,2,1)$, meaning that $\omega = 6$. For parties $1$ and $2$, observe that $\quota(1) = \ell^{o}(1) = 3$ and $\quota(2) = \ell^{o}(2) = 2$ hold, respectively, while for party~$3$, we have $\quota(3) = 1$ and $\ell^{o}(3) = 0$. Now, consider the seat assignment $\vec{s} = (3,2,1)$. This clearly does not satisfy $\WLQo$ with $\rep{1}{\vec{s}} = 1 < \ell^{o}(1) = 3$ but note that the addition of any of the two seats that party~$1$ did not get assigned would suffice in helping it cross $\quota(1) = 3$, so WLQ-X is satisfied.
\end{example}

\noindent
We continue with the proof of Proposition~\ref{prop:WLQ_X_exists_m=2}, which states that, for two-party elections, we can always find a seat assignment that providess $\WLQo$.

\begin{proof}[Proof of Proposition~\ref{prop:WLQ_X_exists_m=2}]
Recall that the weight vector $\vec{w} = (w_{1},\ldots,w_{k})$ needs to be non-increasing (so $w_k$ is minimal).
We devise a method to find a seat assignment $\vec{s}$ that provides WLQ-X: 
let $t \in [k]$ be the minimal value such that $\sum^k_{i = t} w_{i} > \quota(1)$, and assign seats $t+1$ to $k$ to party~$1$ (so $\rep{1}{\vec{s}} \leqslant \quota(1)$).
Then assign the remaining seats to party~$2$ to obtain seat assignment $\vec{s}$. 

Observe that amongst the seats that party~$1$ was not assigned, seat $t$ has the lowest weight $w_{t}$. Moreover, we have $\rep{1}{\vec{s}} + w_{t} > \quota(1)$. Hence, WLQ-X is satisfied with respect to party $1$.
Now, let us assess the situation for party~$2$. Since $\rep{1}{\vec{s}} \leqslant \quota(1)$ holds by the choice of $t$ and we know $\quota(1) + \quota(2) = \rep{1}{\vec{s}} +\rep{2}{\vec{s}} = \omega$ (as $\vec{s}$ is complete), it must be the case that $\rep{2}{\vec{s}} \geqslant \quota(2)$. Hence, WLQ-X is also satisfied with respect to party~$2$.
\end{proof}

\noindent
Here is an example illustrating the failure of WLQ-1 for $\Adams$. 

\begin{namedexample}[$\Adams$ fails WLQ-1]\label{Adams-fails-WLQ-1}
Consider a scenario with four parties, $\prof{v} = (9,1,1,1)$, and $\vec{w} = (1,1,1,1)$.
Then, the weighted quota of party $1$ is $\quota(1) = 4\cdot \nicefrac{9}{12} = 3$. However, $\Adams$ gives each party $p\in[4]$ an initial ratio of $\textit{ratio}_p = \infty$, so each party receives exactly one seat.
Consequently, party~$1$ is two seats away from its weighted lower quota.
\end{namedexample}

\subsection{\new{Supplementary} Material for Section~\ref{sec:upper-quota}}

First, let us discuss the notion of \emph{obtainable upper quota} we had skimmed over in the body of the paper.

In contrast to the obtainable lower quota, we do not incorporate
an `upper quota of seats' in defining a weighted upper quota for party $p$ that is obtainable.
\new{The reason is that} there are election instances where the sum of obtainable weighted upper quotas is less than the total weight,
and, therefore, any reasonable upper-quota axiom based on this notion is bound to be violated. 

\begin{example}
Consider a scenarios with five parties with $\prof{v} = (40,15,15,15,15)$ and $\vec{w} = (5,1,1,1,1,1)$.
Then party $1$ has an upper quota of seats of $2$ seats but and an upper quota of weight of less
than $5$. All other parties have an upper quota of seats of just $1$. Thus, the obtainable upper quotas taking the `upper quota of seats' into
account would be $2$ for party $1$ and $1$ for the other parties, which sums to $7 < 10 = \omega$.  
\end{example}

\noindent
We thus define the obtainable weighted upper quota as follows:
\[
    u^o(p) = \min\left\{r\in R(\vec{w})_{[k]} \mid r \geqslant \quota(p)\right\}.
\]

\noindent
Now we can define the axiom $\WUQo$ based on $u^o$.

\begin{definition}[Obtainable Weighted-seat Upper Quota, $\WUQo$]
    A seat assignment $\vec{s}$ provides $\WUQo$ if for every party $p$ it is the case that $\rep{p}{\vec{s}} \leqslant u^o(p)$.
\end{definition}

\noindent
Note that for unit-weight election instances, $\WUQo$ reduces to the standard UQ property.
For this definition of the obtainable upper quota, we achieve essentially the same 
results as for $\WLQo$. The following four results can be seen as direct counterparts to those we obtained for $\WLQo$. 

\begin{proposition}\label{prop:WUQ_exists_2_parties}
    For every election instance with two parties there exists a seat assignment 
    that satisfies $\WUQo$.
\end{proposition}

\begin{proof}
Consider an election instance with parties $1$ and $2$.
By the definition of the obtainable weighted upper quota, there exists a set of seats $T$ such that $\sum_{t\in T} w(t) = u^o(1)$. Now, let $\vec{s}$ be a seat assignment such that
party~$1$ gets all seats in $T$, and party~$2$ all other seats. By definition,
$\rep{1}{\vec{s}} = u^o(1)$ for party~$1$. Moreover, for party~$2$, we have: 
\begin{align*}
    \rep{2}{\vec{s}} & = \omega - \rep{1}{\vec{s}} = \omega - u^o(1)\\
                       & \leqslant \omega - \omega\cdot\frac{v_{1}}{n} =  \omega\cdot\frac{n - v_{1}}{n} =
                         \omega\cdot\frac{v_{2}}{n} \leqslant u^o(2).
\end{align*}
Hence, $\WUQo$ is satisfied.
\end{proof}

\begin{proposition}\label{prop:WUQ_may_not_exist}
    There are election instances where no complete seat assignment provides~$\WUQo$.
\end{proposition}

\begin{proof}
    We use the election instance familiar from the proof of Proposition~\ref{prop:WLQ_may_not_exist}, where there are three parties, a weight vector $\vec{w} = (3,2,1)$, and a vote vector $\prof{v} = (1,1,1)$. 
    Then each party $p\in [3]$ has an obtainable weighted upper quota of $u^o(p) = 2$. However, providing each party with at most $2$ in representation cannot be achieved if we need to assign all three seats. 
\end{proof}

\begin{proposition}\label{prop:WUQ_is_NP-hard_to_find}
If there exists a polynomial-time algorithm $\alpha$ that finds a seat assignment $\vec{s}$ that provides $\WUQo$ whenever such a seat assignment exists, then $\P=\NP$.
\end{proposition}

\begin{proof}
    Consider the reduction from the $\textsc{Partition}$ problem in the proof of Proposition~\ref{prop:WLQ_is_NP-hard_to_find}. Observe that in the election instance constructed in this reduction, the obtainable weighted lower quota $\ell^o(p)$ for each party $p\in[2]$ is exactly equal to the obtainable weighted upper quota $u^o(p)$. So the same argument provided for $\WLQo$ also works for~$\WUQo$.
\end{proof}

\begin{proposition}\label{prop:WUQ_pseudopoly_to_find}
    Given a constant number of parties and the weights in $\vec{w}$ being polynomial in the input size, finding a seat assignment $\vec{s}$ that provides $\WUQo$ can be done in polynomial time, assuming such a seat assignment exists.
\end{proposition}

\begin{proof}
    Consider the dynamic programming algorithm from the proof of Proposition~\ref{prop:WLQ_pseudopoly_to_find}. Observe that the same algorithm can be deployed to find $\WUQo$-providing seat assignments with the following simple modification: once $\mathcal{W}_{k}$ is computed, the algorithm checks each tuple $(W_{1},\ldots,W_{m})$ to determine whether, for every $p\in[m]$, it holds that $W_{p}\leqslant u^{o}(p)$ instead of whether $W_{p}\geqslant \ell^{o}(p)$. And we can compute $u^o(p)$ in polynomial time due to $(i)$ the assumption on the weights in $\vec{w}$, and $(ii)$ observing this task is also equivalent to solving the \textsc{SubsetSum} problem. 
\end{proof}

\noindent


\noindent
How do WUQ-1 and WUQ-X relate to $\WUQo$? Here we \new{uncover a} difference between upper and lower quota: 
$\WUQo$ \new{implies not only} WUQ-1, but WUQ-X \new{as well}.

\begin{proposition}\label{prop:WUQo_implies_WUQ-X}
 $\WUQo$ implies WUQ-X.
 \end{proposition}


\begin{proof} 
    Assume WUQ-X is violated by a seat assignment $\vec{s}$. Then there exists a party~$p$ such that for every seat $t\in \vecPos{p_{x}}{\vec{s}}$, it holds that $\rep{p}{\vec{s}} > \rep{p}{\vec{s}} - w_{t} \geqslant \quota(p)$. But then this means that $\rep{p}{\vec{s}} - w_{t}$ is an achievable weight representation value at least as large as $\quota(p)$, i.e., $\rep{p}{\vec{s}} - w_{t} \in \{r\in R(\vec{w})_{[k]}:r\geqslant \quota(p)\}$. Since $u^o(p)$ is the smallest of the weights in $\{r\in R(\vec{w})_{[k]}:r\geqslant \quota(p)\}$, we get $\rep{p}{\vec{s}} > \rep{p}{\vec{s}} - w_{t} \geqslant u^o(p)$ and hence, $\WUQo$ is also violated.
\end{proof}

\noindent
Here is an example that shows that WUQ-X does not imply $\WUQo$.

\begin{example}
   Consider three parties, a vote vector $\prof{v} = (3,2,1)$, and a weight vector $\vec{w} = (3,2,1)$, and (so $\omega = 6$). See that $\quota(1) = u^{o}(1) = 3$, $\quota(2) = u^{o}(2) = 2$ and $\quota(3) = u^{o}(3) = 1$ for the parties. Now, take a seat assignment $\vec{s} = (3,2,1)$. This does not satisfy $\WUQo$ as we get that $\rep{3}{\vec{s}} = 3 > u^{o}(3) = 1$ but note that the removal of the seat that party~$3$ was assigned, would suffice in helping it go below $\quota(3) = 1$. 
\end{example}

\noindent
\begin{newenv}
Here is an example that shows that WEF1 does not imply WUQ-X.

\begin{example}\label{exp:WEF1-WUQ-X}
   Consider two parties, a vote vector $\prof{v} = (1,1)$, and a weight vector $\vec{w} = (11,2,1)$, and (so $\omega = 14$). See that $\quota(1) = \quota(2) = 7$. Now, take a seat assignment $\vec{s} = (1,2,1)$. This does not satisfy WUQ-X as we get that $\rep{1}{\vec{s}} - w_3 = 12-1 > 7$ but note that the removal of seat $1$ from party~$1$ would remove any envy from party~$2$. 
\end{example}
\end{newenv}

\noindent
The following example illustrates the well-known fact that D'Hondt violates UQ already in the unit-weight setting by an arbitrary number of seats \cite{BalinskiYoung1982}. This example doubles as an instance where WWEF1 is satisfied and thereby shows that WWEF1 does not imply WUQ-1.

\begin{namedexample}[$\DHondt$ fails WUQ-1 \new{and WWEF1 does not imply WUQ-1}]\label{exm:dhont_fails_WUQ-1}
Consider 101 parties, $\prof{v} = (100,1,\ldots,1)$, and $\vec{w} = (1,1,1,1)$. Note that $\quota(1) = 2$ and \new{$\quota(p) = 0.02$} for all $p\in\{2,\ldots,101\}$. $\DHondt$ assigns all seats to party~$1$ which violates WUQ-1. 

\new{We will now show that while this $\DHondt$ seat assignment $\vec{s} = (1,1,1,1)$ fails WUQ-1, it satisfies WWEF1. Observe that it satisfies WWEF1 since party~$1$ exceeded its quota and the second condition of WWEF1 is met as for each party $p\in\{2,\ldots,101\}$, we can take a weight-$1$ seat from party~$1$ such that $\nicefrac{ \rep{p}{\vec{s}} + 1}{v_{p}} = 1 > 0.0.4 = \nicefrac{4}{100} = \nicefrac{\rep{1}{\vec{s}}}{v_{1}}$ holds, and we have that $\nicefrac{\rep{p_{x}}{\vec{s}}}{v_{p_{x}}} = 0 = \nicefrac{\rep{p_{y}}{\vec{s}}}{v_{p_{y}}}$ for each $p_{x},p_{y}\in \{2,\ldots,101\}$.}
\end{namedexample}

\noindent
Finally, we supply the omitted proof for Proposition~\ref{Prop:WEFX-WUQ-X}, claiming that WEFX implies WUQ-X.

    \begin{proof}[Proof of Proposition~\ref{Prop:WEFX-WUQ-X}]
        \new{O}bserve that if WUQ-X is violated, then there exists a party~$p_{x}$ such that 
        $(\rep{p_{x}}{\vec{s}} - w_{t}) \geqslant \quota(p_{x}) =\omega \cdot\nicefrac{v_{p_{x}}}{n}$ for some $t\in \vecPos{p_{x}}{\vec{s}}$. 
        On the other hand, for WEFX to hold, for every party $p_{y}\in[m] \setminus\! \{p_{x}\}$ and every $t\in \vecPos{p_{x}}{\vec{s}}$ it must be the case that
        \(
            \nicefrac{\rep{p_{y}}{\vec{s}}}{v_{p_{y}}} \geqslant \nicefrac{(\rep{p_{x}}{\vec{s}} - w_{t})}{v_{p_{x}}}.
        \)
        \new{T}hese \new{inequalities} imply that $\rep{p}{\vec{s}} \geqslant \omega \cdot \nicefrac{v_{p}}{n} = \quota(p)$ for every party $p\in[m]$, which is not possible if party~$p_{x}$ exceeded $\quota(p_{x})$.
    \end{proof}

\subsection{Note on the Full-Supply Assumption}

We touch here on the assumption that each party has enough members to fill 
all available \(k\) seats, introduced in Section~\ref{sec:model} under the name of \emph{full supply}.
Specifically, we want to make the case that it is a necessary assumption. 

This is important given that a natural follow-up to our results on apportionment 
would be to lift the positive weighted-seat results to a more general setting, e.g., multiwinner voting,
and one step in this direction would be to forego full supply.
However, we find that, once making this step, even the weakest of our axioms fail to be satisfied.

\begin{namedexample}[WLQ-1 without full supply]
    Consider four parties, $\prof{v} = (3,3,3,1)$ and $\vec{w} = (15,15,1,\ldots,1)$ with $k=72$. Suppose each party $p\in [3]$ can receive two seats while party~$4$ can receive $66$ seats. So some party~$p\in [3]$ with $\quota(p) = 30$ must be assigned two weight-$1$ seats, so WLQ-1 cannot be provided.
\end{namedexample}

\begin{namedexample}[WUQ-1 without full supply]
    Consider three parties, $\prof{v} = (1,1,1)$ and $\vec{w} = (3,3,3,3,3)$. Suppose each party $p\in [2]$ can receive one seat while party~$3$ can receive three seats. So, party~$3$ with $\quota(p) = 3$ must be assigned three seats of weight $3$, which violates WUQ-1. 
\end{namedexample}

\noindent
These examples illustrate that extra care needs to be taken when generalising our work on apportionment with weighted seats to more complex settings (but see \cite[Chapter~5]{ChingomaPhD2025}).  

\end{document}